\documentclass[11pt]{article}

\title{Harmonic Analysis of Boolean Networks: Determinative Power and Perturbations}

\author{Reinhard Heckel\thanks{Department of Information Technology and Electrical Engineering, ETH Zurich}, \; 
Steffen Schober\thanks{Institute of Telecommunications and Applied Information Theory, University of Ulm} \;
and Martin Bossert${}^{\dagger}$
}

\date{}
\usepackage{setspace}

\renewcommand{\vspace}[1]{}

\date{}
\usepackage{setspace}

\usepackage{amssymb}

\usepackage[figuresright]{rotating}

\usepackage{amsmath,amssymb}
\usepackage{graphicx}
\usepackage{dcolumn}
\usepackage{bm}
\usepackage{amscd,amsthm}
\usepackage{ifthen}

\usepackage{mathtools} 

\usepackage[applemac]{inputenc}

\usepackage{tikz}
\usepackage{pgfplots}
\usepackage{subfigure}
\usetikzlibrary{pgfplots.groupplots}

\usepackage{bm}

\usepackage{hyperref}
\hypersetup{
    colorlinks=false,       
    linkcolor=red,          
    citecolor=green,        
    filecolor=magenta,      
    urlcolor=cyan           
}

\newcommand\vect[1]{\ensuremath{\mathbf{ #1}}}
\newcommand\XA[0]{\mathbf{X}_A }
\newcommand\xa[0]{\mathbf{x}_A }
\newcommand\PR[1]{\ensuremath{ {\mathrm{P}\!}\left[#1\right]}}
\newcommand\MI[2]{\mathrm{MI}(#1 ; #2 )}
\newcommand\INF[2]{I_{#1}(#2)}
\newcommand\PFC[0]{\hat f }
\newcommand\OBF[0]{ \Phi}

\newcommand\Tex{}
\newcommand\EX[2][\Tex]{
\ifthenelse{\equal{#1}{}}{{\mathbb E}[#2]}{\ensuremath{\underset{#1}{\mathbb E}\left[ #2\right]}}}
\newcommand\VAR[1]{{\mathrm{Var}}\left(#1\right)}
\newcommand\IPP[2]{ \langle #1 , #2 \rangle } 
\newcommand\AV[1]{\mathrm{as}(#1)}

\newcommand\defeq{\triangleq}

\newcommand\norm[2][\Tnorm]{\ensuremath{{\|#2\|}_{#1}}}

\newcommand\ind[1]{\ensuremath{ \mathbf 1_{\{ #1 \}} }}

\newtheorem{lemma}{Lemma}
\newtheorem{corollary}{Corollary}
\newtheorem{theorem}{Theorem}

\newtheorem{proposition}{Proposition}

\theoremstyle{bmcstyleone}
\newtheorem{definition}{Definition}

\begin{document}
\maketitle

\begin{abstract}
Consider a large Boolean network with a feed forward structure. Given a probability distribution on the inputs, can one find, possibly small, collections of input nodes that determine the states of most other nodes in the network? To answer this question, a notion that quantifies the \textit{determinative power} of an input over the states of the nodes in the network is needed. We argue that the mutual information (MI) between a given subset of the inputs $\vect{X} = \{X_1, ...,\allowbreak X_n\}$ of some node $i$ and its associated function $f_i(\vect{X})$ quantifies the determinative power of this set of inputs over node $i$. We compare the determinative power of a set of inputs to the sensitivity to perturbations to these inputs, and find that, maybe surprisingly, an input that has large sensitivity to perturbations does not necessarily have large determinative power. However, for \textit{unate} functions, which play an important role in genetic regulatory networks, we find a direct relation between MI and sensitivity to perturbations. As an application of our results, we analyze the large-scale regulatory network of \textit{Escherichia coli}. We identify the most determinative nodes and show that a small subset of those reduces the overall uncertainty of the network state significantly. Furthermore, the network is found to be tolerant to perturbations of its inputs.
\end{abstract}

\section{Introduction}
A Boolean network (BN) is a discrete dynamical system, which is, for example, used to study and model a variety of biochemical networks such as genetic regulatory networks. BNs have been introduced in the late 1960s by Kauffman \cite{kauffman_metabolic_1969,kauffman_homeostasis_1969} who proposed to study random BNs as models of gene regulatory networks. Kauffman investigated their dynamical behavior and a phenomena called self-organization. Aside from its original purpose, BNs were also used to model (small-scale) genetic regulatory networks; for example, in \cite{davidich_boolean_2008,saez-rodriguez_logical_2007,li_yeast_2004}, it was demonstrated that BNs are capable of reproducing the underlying biological processes (i.e., the cell cycle) well. BNs are also used to model large-scale networks, such as the \textit{Escherichia coli} regulatory network \cite{covert_integrating_2004} which is analyzed in section~\ref{sec:ecoli}. This network is, in contrast to Kauffman's automata and the regulatory networks considered in \cite{davidich_boolean_2008,saez-rodriguez_logical_2007,li_yeast_2004}, not an autonomous system, since the gene's states are determined by external factors.

In the literature addressing the analysis of BNs, it is common to consider measures that quantify the effect of perturbations. Whether a random BN operates in the so called ordered or disordered regime is determined by whether a single perturbation, i.e., flipping the state of a node, is expected to spread or die out eventually. Kauffman \cite{kauffman_homeostasis_1969} argues that biological networks must operate at the border of the ordered and disordered regime; hence, they must be tolerant to perturbations to some extent.

In contrast to measures of perturbations, determinative power in BNs has not received much attention, even though there are several settings where such a notion is of interest. For example, given a feed forward network where the states of the nodes are controlled by the states of nodes in the input layer, we might ask whether a possibly small set of inputs suffices to determine most states, i.e., reduces the uncertainty about the network's states significantly. This can be addressed by quantifying the determinative power of the input nodes. For example, in the \textit{E. coli} regulatory network, it turns out that a small set of metabolites and other inputs determine most genes that account for \textit{E. coli}'s metabolism (see section~\ref{sec:ecoli}).

In this paper, we view the state of each node in the network as an independent random variable. This modeling assumption applies for networks with a tree-like topology, e.g., a feed forward network, and is often applied when studying the effect of perturbations. For this setting, determinative power of nodes and perturbation-related measures are properties of single functions; hence, the analysis of the BN reduces to the analysis of single functions. Our main tool for the theoretical results is Fourier analysis of Boolean functions. Fourier analytic techniques were first applied to BNs by Kesseli et al.  \cite{kesseli_spectral_2005,kesseli_tracking_2005}. In \cite{kesseli_spectral_2005,kesseli_tracking_2005}, results related to Derrida plots and convergence of trajectories in random BNs were derived. Ribeiro et al. \cite{ribeiro_mutual_2008} considered the pairwise mutual information in time series of random BNs, under a different setup that we use. Specifically, in \cite{ribeiro_mutual_2008}, the functions are random; whereas here, the functions are deterministic, but the argument is random. Finally, note that part of this paper was presented at the 2012 International Workshop on Computational Systems Biology \cite{heckel_determinative_2012}.

\subsection{Contributions}
Mutual information between a set of inputs to a node and the state of this node is a measure of the determinative power of this set of inputs, as mutual information quantifies mutual dependence of random variables. In order to understand the determinative power and mutual dependencies in Boolean networks, we systematically study the mutual information of sets of inputs and the state of a node. We relate mutual information to a measure of perturbations and prove that (maybe surprisingly) a set of inputs that is highly sensitive to perturbations might not necessarily have determinative power. Conversely, a set of inputs which has determinative power must be sensitive to perturbations. To prove those results, we show that the concentration of weight in the Fourier domain on certain sets of inputs characterizes a function in terms of tolerance to perturbations and determinative power of input nodes. Furthermore, we generalize a result by Xiao and Massey \cite{xiao_spectral_1988}, which gives a necessary and sufficient condition of statistical independence of a set of inputs and a function's output in terms of the Fourier coefficients. This result can for instance be applied to decide for which classes of functions the algorithm presented in \cite{liang_reveal_1998}, which detects functional dependencies based on estimating mutual information, can succeed or fails. For \textit{unate} functions, we show that any input and the function's output are statistically dependent and provide a direct relation between the mutual information and the influence of a variable. The class of unate functions is especially relevant for biological networks, as it includes all linear threshold functions and all nested canalizing functions, and describes functional dependencies in gene regulatory networks well \cite{grefenstette_analysis_2006}. As an application of the theoretical results in this paper, we show that mutual information can be used to identify the determinative nodes in the large-scale model of the control network of \textit{E. coli}'s metabolism \cite{covert_integrating_2004}.

\subsection{Outline}
The paper is organized as follows. Boolean networks and Fourier analysis of Boolean functions are reviewed in section~\ref{sec:preliminaries}. In section~\ref{sec:inf}, the influence and average sensitivity as measures of perturbations are reviewed, and their relation to the Fourier spectrum is discussed. In section~\ref{sec:mi}, we study the mutual information of sets of inputs and the function's output. Section~\ref{sec:unate} is devoted to unate functions. Section~\ref{sec:ecoli} contains an analysis of the large-scale \textit{E. coli} regulatory network, using the tools and ideas developed in previous sections.

\section{Preliminaries}\label{sec:preliminaries}
\newcommand{\N}{n}

We start with a short introduction to Boolean networks and Fourier analysis of Boolean functions, and introduce notation.

\subsection{Boolean networks}
A (synchronous) BN can be viewed as a collection of $\N$ nodes with memory. The state of a node $i$ is described by a binary state $x_{i}(t) \in \{-1,+1\}$ at discrete time $t \in \mathbb N$. Choosing the alphabet to be $\{-1,+1\}$ rather than $\{0,1\}$ as more common in the literature on BNs will turn out to be advantageous later. However, both choices are equivalent. The state of the network at time $t$ can be described by the vector $\vect{x}(t)= [x_1(t),...,x_n(t)] \in \{-1,+1\}^\N$. The network dynamic is defined by 
\begin{equation}\label{eq:updating}
x_{i}(t+1) = f_{i}(\vect{x}(t)),
\end{equation}
where $f_{i}\colon \{-1,+1\}^\N \to \{-1,+1\}$ is the Boolean function associated with node $i$. At time $t=0$, an initial state $\vect{x}(0)=\vect{x}_0$ is chosen. In general, not all arguments $x_1,...,x_\N$ of a function $f_{i}(\vect{x})$ need to be \textit{relevant}. The  variable $x_j, j\in \{1,...,\N\}$ is said to be relevant for $f_{i}$ if there exists at least one $\vect{x} \in \{-1,+1\}^\N$, such that changing $x_j$ to $-x_j$ changes the function's value. In most of the BN models in biology, the functions depend on a small subset of their arguments only. Furthermore, not every state must have a function associated with it; states can also be external inputs to the network.

To study the determinative power and tolerance to perturbations, a probabilistic setup is needed. In our analysis, we assume that each state is an independent random variable $X_{i}$ with distribution $\PR{X_{i} = x_{i} }$, $x_i\in \{-1,+1\}$. The assumption of independence holds for networks with tree-like topology, but is not feasible for networks with strong local dependencies and feedback loops. However, in many relevant settings, a BN has a tree-like topology, for instance the \textit{E. coli} network analyzed in section~\ref{sec:ecoli}. For a network with few local dependencies, assuming independence will lead to a small modeling error. Major results concerning the analysis of BNs have been obtained under the assumptions as stated above, e.g., the annealed approximation \cite{derrida_random_1986}, an important result on the spread of perturbations in random BNs. Several important results on random BNs, e.g., \cite{derrida_random_1986}, let the network size $n$ tend to infinity; hence, there are no local dependencies.

\subsection{Notation}

We use $[n]$ for the set $\{1,2,...,n\}$, and all sets are subsets of $[n]$. With $\sum_{S\subseteq A}\left( \cdot \right)$, we mean the sum over all sets $S$ that are subsets of $A$.
Throughout this paper, we use capital letters for random variables, e.g., $X$, and lower case letters for their realizations, e.g., $x$. Boldface letters denote vectors, e.g., $\vect{X}$ is a random vector, and $\vect{x}$ its realization. For a vector $\vect{x}$ and a set $A\subseteq [n]$, $\vect{x}_A$ denotes the subvector of $\vect{x}$ corresponding to the entries indexed by $A$.

\subsection{Fourier analysis of Boolean functions }\label{sec:fanalysisintro}

In the following, we give a short introduction to Fourier analysis of Boolean functions. Let $\vect{X} = (X_1, ..., X_n)$ be a binary, product distributed random vector, i.e., the entries of $\vect{X}$ are independent random variables $X_i, i \in [n]$ with distribution $\PR{X_i = x_i}, x_i \in \{-1,+1\}$. Throughout this paper, probabilities $\PR{\cdot}$ and expectations $\EX{\cdot}$ are with respect to the distribution of $\vect{X}$. We denote $p_i \defeq\PR{X_i=1}$, the variance of $X_i$ by $\VAR{X_i}$, its standard deviation by $\sigma_i\defeq\sqrt{\VAR{X_i}}$ and finally $\mu_i\defeq\EX{X_i}$. The inner product of $f,g: \{-1,+1\}^n \to \mathbb \{-1,+1\}$ with respect to the distribution of $\vect{X}$ is defined as
\begin{equation}
\IPP{f}{g} \defeq \EX{f(\vect{X}) g(\vect{X})} = \!\!\!\! \sum_{\vect{x} \in \{-1,1\}^n} \!\!\! \PR{\vect{X}=\vect{x}}  f(\vect{x}) g(\vect{x})
\label{eq:inner_prod}
\end{equation}
which induces the norm $\norm{f} = \sqrt{\IPP{f}{f}}$.
An orthonormal basis with respect to the distribution of $\vect{X}$ is
\begin{equation}
\OBF_{S}(\vect{x}) = \prod_{i \in S} \frac{x_i-\mu_i}{\sigma_i}, \quad S \subseteq [n] \setminus \emptyset
\label{eq:basisfunc}
\end{equation}
and
\[
\OBF_{S}(\vect{x}) = 1, \quad S = \emptyset.
\]
This basis was first proposed by Bahadur \cite{bahadur_representation_1961}.
Thus, each Boolean function $f\colon\{-1,+1\}^n \to  \{-1,+1\}$ can be uniquely expressed as
\begin{equation}
f(\vect{x}) = \sum_{S \subseteq [n]} \PFC (S) \OBF_S(\vect{x}),
\label{eq:fourier_expansion}
\end{equation}
where $\PFC (S) \defeq \IPP{f}{\OBF_S}$ are the Fourier coefficients of $f$. Note that \eqref{eq:fourier_expansion} is a representation of $f$ as a multilinear polynomial. As an example, consider the AND2 function defined as $f_{\mathrm{AND2}}(\vect{x}) = 1$ if and only if $x_1=x_2=1$, and let $p_1=p_2=1/2$. According to \eqref{eq:fourier_expansion}
\[
    f_{\mathrm{AND}}(\vect{x}) = -\frac{1}{2} + \frac{1}{2}x_1  + \frac{1}{2}x_2 + \frac{1}{2}x_1 x_2.
\]
As a second example consider PARITY2, i.e., the XOR function, defined as $f_{\mathrm{PARITY2}}(\vect{x}) = 1$ if $x_1=x_2=1$ or if $x_1=x_2=-1$, and  $f_{\mathrm{PARITY2}}(\vect{x}) = -1$ for all other choices of $\vect{x}$. Written as a polynomial, $f_{\mathrm{PARITY2}}(\vect{x}) = x_1 x_2$. We conclude this section by listing properties of the basis functions which are used frequently throughout this paper.

 {\bf Decomposition:} Let $A\subseteq [n]$ and $S \subset A$, and denote $\bar S = A \setminus S$. Then,
\[
\OBF_A(\vect{x}) = \OBF_S(\vect{x})  \OBF_{\bar S}(\vect{x}).
\]

{\bf Orthonormality:} For $A,B \subseteq [n]$,
\[
\EX{\OBF_A(\vect{X})  \OBF_B(\vect{X})} =
 \begin{cases} 1, \text{if } A=B \\
               0, \text{otherwise}.
\end{cases}
\]

{\bf Parseval's identity:} For $f\colon \{-1,+1\}^n \to \{-1,+1\}$,
\[
\EX{f(\vect{X})^2} = \norm{f}^2 = \sum_{S \subseteq [n]} \PFC (S)^2 = 1.
\]

\section{Influence and average sensitivity}\label{sec:inf}

Next, we discuss measures of perturbations and their relation to the Fourier spectrum. We start with a measure of the perturbation of a single input.

\begin{definition}[{\cite{ben-or_collective_1985}}]
Define the influence of variable $i$ on the function $f$ as
\[
I_i(f) = \PR{f(\vect{X}) \neq f(\vect{X} \oplus e_i)},
\]
where $\vect{x}\oplus e_i$ is the vector obtained from $\vect{x}$ by flipping its $i$th entry.
\end{definition}

By definition, the influence of variable $i$ is the probability that perturbing, i.e., flipping, input $i$ changes the function's output. Influence can be viewed as the capability of input $i$ to change the output of $f$. In BNs, usually, the sum of all influences, i.e., the \textit{average sensitivity} is studied.

\begin{definition} The average sensitivity of $f$ to the variables in the set $A$ is defined as
\[
\INF{A}{f} = \sum_{i \in A} I_i(f).
\]
The average sensitivity of $f$ is defined as $\AV{f} \defeq \INF{\{1,...,n\}}{f}$.
\end{definition}

$\INF{A}{f}$ captures whether flipping an input chosen uniformly at random from $A$ affects the function's output. Most commonly, all inputs are taken into account, i.e., the average sensitivity $\AV{f}$ is studied. As an example, $\AV{f_{\mathrm{PARITY2}}} = 2$ and $\AV{f_{\mathrm{AND2}}} = 1$; hence, PARITY2 is more sensitive to single perturbations than AND2. Influence and average sensitivity have the following convenient expressions in terms of Fourier coefficients.

\begin{proposition}[Lemma 4.1 of \cite{bshouty_fourier_1996}]
 For any Boolean function $f$,
\begin{align}
 \INF{i}{f} = \frac{1}{\sigma_i^2} \sum_{S\subseteq [n]\colon i\in S} \PFC(S)^2.
\label{eq:infinteroffcoef}
\end{align}
\label{prop:influenc_spectrum}
\end{proposition}

\begin{proposition}
For any Boolean function $f$,
\begin{equation}
 \INF{A}{f} = \sum_{S\subseteq [n]} \PFC(S)^2 \sum_{ i\in S \cap A } \frac{1}{\sigma_i^2}.
\label{eq:av_fourier}
\end{equation}
\label{pr:av_fourier}
\end{proposition}

Proposition~\ref{pr:av_fourier} follows directly from Proposition~\ref{prop:influenc_spectrum} and the definition of $\INF{A}{f}$. From \eqref{eq:av_fourier}, we see that $\AV{f}$ is large if the Fourier weight is concentrated on the coefficients of high degree $d= |S|$, i.e., if $\sum_{S \colon |S| \geq d} \PFC(S)^2$ is large (i.e., close to one). For this case, Parseval's identity implies that the $\PFC(S)^2$ with $|S| < d$ must be small. Let's see an example: Suppose $p_1=p_2=p_3=1/2$ and consider the AND3 function, i.e., $f_{\mathrm{AND3}}(x_1,x_2,x_3) = 1$ if and only if $x_1=x_2=x_3=1$.  $f_{\mathrm{AND3}}$ is tolerant to perturbations since $\AV{f_{\mathrm{AND3}}}=0.75$, and as Figure \ref{fig1} shows, its spectrum is concentrated on the coefficients of low degree. In contrast for $f_{\mathrm{PARITY3}}(x_1,x_2,x_3) \defeq x_1 x_2 x_3$, $\AV{f_{\mathrm{PARITY}}}=3$. Hence, PARITY3 is maximally sensitive to perturbations. Figure \ref{fig1} shows that its spectrum is maximally concentrated on the coefficient of highest degree.

\begin{figure}
\begin{center}
\includegraphics[width=0.35\textwidth]{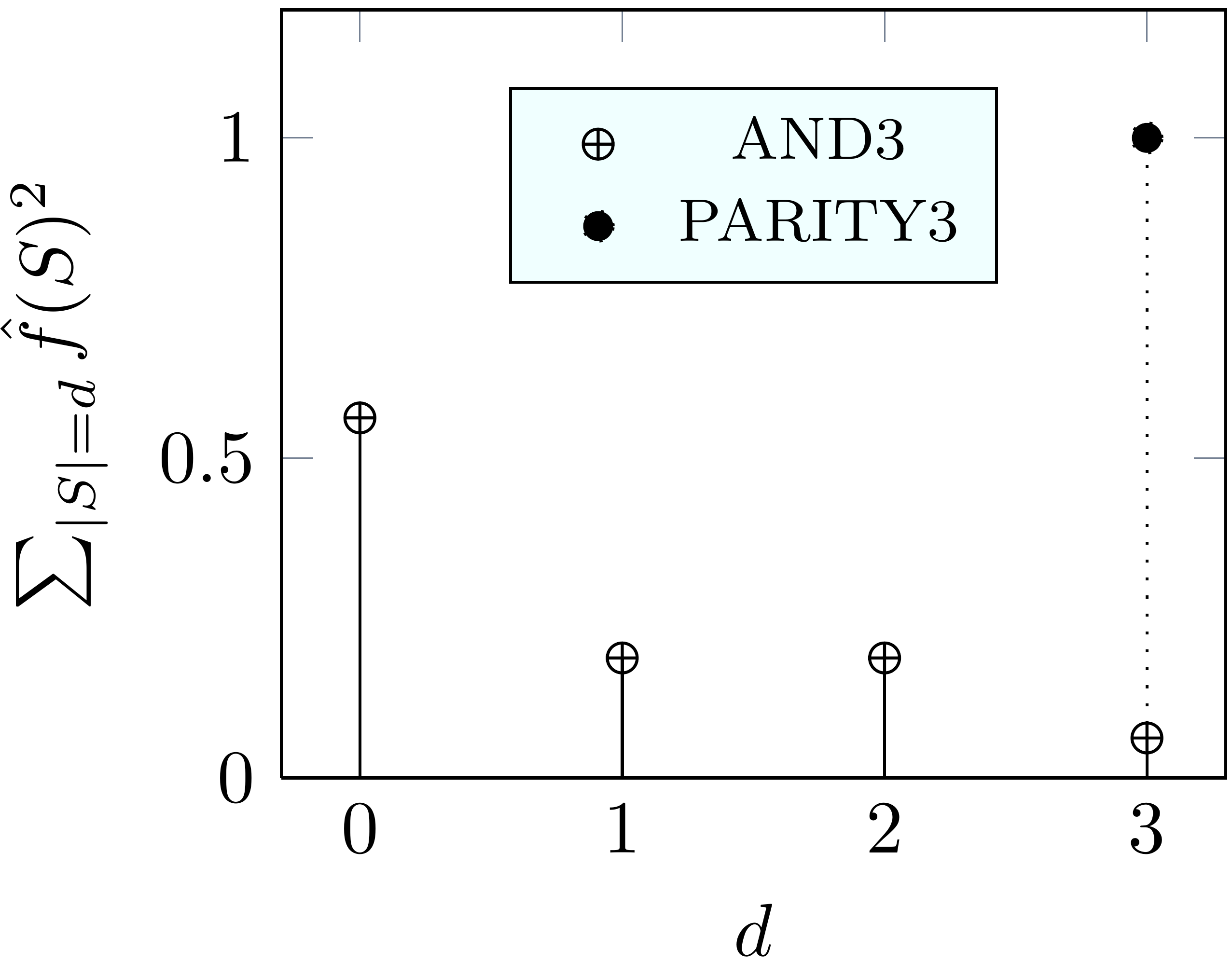}
\end{center}
\caption{The Fourier spectrum of AND3 and PARITY3.} \label{fig1} 
\end{figure}

According to \eqref{eq:av_fourier} $\AV{f}$ is small only if the Fourier weight is concentrated on the coefficients of low degree. This is the case either if $f$ is strongly biased (i.e., if $f(\vect{x}) = a$, for most inputs $\vect{x}$, where  $a \in \{-1,1\}$ is a constant) or if $f$ depends on few variables only.
This is in accordance with the results of Kauffman \cite{kauffman_metabolic_1969}; he found that a random BN operates in the ordered regime if the functions in the network depend on average on few variables.

We will state our result for measures of single perturbations. However, these results also apply to other noise models, specifically to the \textit{noise sensitivity} of $f$. That is, because the noise sensitivity of $f$ is small if $f$ is tolerant to single perturbations. The noise sensitivity of a Boolean function is defined as the probability that the function's output changes if each input is flipped independently with probability $\epsilon$. For uniformly distributed $\vect{X}$, $\epsilon \,\AV{f}$ is an upper bound for the noise sensitivity; for small values of $\epsilon$, $\epsilon \, \AV{f}$ approximates the noise sensitivity well. For the $X_i$ being equally but possibly nonuniformly distributed and a slightly different noise model, it was found in \cite{schober_about_2008} that $\epsilon\, \AV{f}$ still upper bounds the noise sensitivity. This result was generalized to product distributed $\vect{X}$ in \cite{matache_sensitivity_2009}.

\section{Mutual information and uncertainty}\label{sec:mi}

In this section, we study the determinative power of a subset of variables $\XA$, where $\XA$ consists of the entries of $\vect{X}$ corresponding to the indices in the set $A\subseteq [n]$, over the function's output $f(\vect{X})$. As a measure of determinative power, we take the mutual information $\MI{f(\vect{X})}{\XA}$ between $f(\vect{X})$ and $\XA$, since $\MI{f(\vect{X})}{\XA}$ quantifies the statistical dependence between the random variable $\XA$ and $f(\vect{X})$. Hence, this section is devoted to the study of $\MI{f(\vect{X})}{\XA}$.

Before giving a formal definition of mutual information, let us start with an example. Consider the PARITY2 function and let its inputs $X_1,X_2$ be uniformly distributed. Intuitively, if $X_1$ has determinative power, knowledge about $X_1$ should provide us with information about $f_{\mathrm{PARITY2}}(\vect{X})$. Suppose we know the value of $X_1$, say $X_1 = 1$. Since $f_{\mathrm{PARITY2}}(\vect{x}) = x_1x_2$, we have with $\PR{X_2=1}=1/2$ that $\PR{f_{\mathrm{PARITY2}}(\vect{X}) = 1} = \PR{f_{\mathrm{PARITY2}} = 1 |X_1=1}$. Hence, knowledge of $X_1$ does not help to predict the value of $f_{\mathrm{PARITY2}}(\vect{X})$. Therefore, $X_1$ has no determinative power over $f_{\mathrm{PARITY2}}(\vect{X})$. We indeed have $\MI{f_{\mathrm{PARITY2}}(\vect{X})}{X_1} = 0$.

We next define mutual information. Mutual information is the reduction of uncertainty of a random variable $Y$ due to the knowledge of  $X$; therefore, we need to define a measure of uncertainty first, which is entropy. As a reference for the following definitions, see \cite{cover_elements_2006}.

\begin{definition}
The entropy $H(X)$ of a discrete random variable $X$ with alphabet $\mathcal X$ is defined as
\[
H(X) \defeq -\sum_{x \in \mathcal X} \PR{X=x} \log_2  \PR{X=x}.
\]
\end{definition}

\begin{definition}
The conditional entropy $H(Y | X)$ of a pair of discrete and jointly distributed random variables $(Y,X)$ is defined as
\[
H(Y|X) \defeq \sum_{x \in \mathcal X} \PR{X=x} H (Y | X=x) .
\]
\end{definition}

\begin{definition}
The mutual information $\MI{Y}{X}$ is the reduction of uncertainty of the random variable $Y$ due to the knowledge of $X$
\[
\MI{Y}{X}\defeq H(Y) - H(Y | X).
\]
\end{definition}

For a binary random variable $X$ with alphabet $\mathcal{X} = \{x_1,x_2\}$ and $p\defeq \PR{X=x_1}$, we have $H(X) = h (p)$, where $h(p)$ is the binary entropy function, defined as
\begin{equation}
h(p) \defeq - p \log_2 p - (1-p) \log_2 (1-p).
\label{eq:BEF}
\end{equation}
The properties of mutual information are what we intuitively expect from a measure of determinative power: If knowledge of $X_i$ reduces the uncertainty of $f(\vect{X})$, then  $X_i$ determines the state of $f(\vect{X})$ to some extent, because then, knowledge about the state of $X_i$ helps in predicting $f(\vect{X})$. Furthermore, we require from a measure of determinative power that not all variables can have large determinative power simultaneously. This is guaranteed for mutual information as
\begin{equation}
\sum_{i=1}^n \MI{f(\vect{X})}{X_i}  \leq \MI{f(\vect{X})}{\vect{X}} \leq 1,
\label{eq:sum_of_mis}
\end{equation}
which follows from the chain rule of mutual information (as a reference, see \cite{cover_elements_2006}) and independence of the $X_i, i\in [n]$. Hence, if $\MI{f(\vect{X})}{X_i}$ is large, i.e., close to 1, we can be sure that $X_i$ has determinative power over $f(\vect{X})$ since \eqref{eq:sum_of_mis} implies that $\MI{f(\vect{X})}{X_j}$ for $j \neq i$  must be small then.

\subsection{Mutual information and the Fourier spectrum}\label{sec:reltov}

In order to study determinative power, its relation to measures of perturbations, and statistical dependencies, we start by characterizing the mutual information in terms of Fourier coefficients.
Our results are based on the following novel characterization of entropy in terms of Fourier coefficients.

\begin{theorem}
Let $f$ be a Boolean function, let $\vect{X}$ be product distributed, and
let $\XA=\{X_i\colon i\in A\}$ be a fixed set of arguments, where $A \subseteq [n]$.
Then,
\[
H(f(\vect{X}) | \XA) = \mathbb E \left[ h\left(\frac{1}{2} \left(1+ \sum_{S \subseteq A} \PFC(S) \OBF_{S}(\XA) \right) \right)  \right],
\]
where $h(\cdot)$ is the binary entropy function as defined in \eqref{eq:BEF}.
\label{th:relation_fourier_H}
\end{theorem}

\begin{proof}
See Appendix \ref{proof:relation_fourier_H}. For the special case of uniformly distributed $\vect{X}$, a proof appears in \cite{forre_methods_1990}, in the context of designing S-boxes.
\end{proof}

Using the definition of mutual information, an immediate corollary of Theorem~\ref{th:relation_fourier_H} is  the following:

\begin{corollary}
Let $f$ be a Boolean function, $\vect{X}$ be product distributed, and  $\XA=\{X_i\colon i\in A\}$. Then,
\begin{equation}
\MI{f(\vect{X})}{\XA}= h\left(1/2(1+ \PFC(\emptyset) \right) \label{eq:michar} - \mathbb E \left[ h\left(\frac{1}{2} \left(1+ \sum_{S \subseteq A} \PFC(S) \OBF_{S}(\XA) \right) \right)  \right].
\end{equation}
\label{co:michar}
\end{corollary}

Theorem~\ref{th:relation_fourier_H} (and Corollary~\ref{co:michar}) shows that the conditional entropy $H(f(\vect{X}) | \XA)$ and the mutual information $\MI{f(\vect{X})}{\XA}$ are functions of the coefficients $\{\PFC(S)\colon S \subseteq A\}$ only. This already hints at a fundamental difference to the average sensitivity, since the average sensitivity depends on the coefficients $\{\PFC(S)\colon |S \cap A|>0\}$, according to \eqref{eq:av_fourier}.

We next discuss $\MI{f(\vect{X})}{X_i}$ based on \eqref{eq:michar}. First, note that $\MI{f(\vect{X})}{X_i}$ has previously been studied under the notion \textit{information gain} as a measure of `goodness' for split variables in greedy tree learners \cite{rosell_why_2005} and also under the notion of \textit{informativeness} to quantify voting power \cite{diskin_voting_2010}. According to \eqref{eq:michar}, the mutual information $\MI{f(\vect{X})}{X_i}$ just depends on $\PFC(\{i\})$, $\PFC(\emptyset)$, and $p_i$. In contrast, the influence $\INF{i}{f}$ is a function of the coefficients $\{ \PFC( S)\colon S\in [n], i\in S \}$, according to \eqref{eq:infinteroffcoef}. In Figure \ref{fig2}, we depict $\MI{f(\vect{X})}{X_i}$ for $p_i=0.3$ as a function of $\PFC(\{i\})$ and $\PFC(\emptyset)$.

\begin{figure}
\begin{center}
\includegraphics[width=0.5\textwidth]{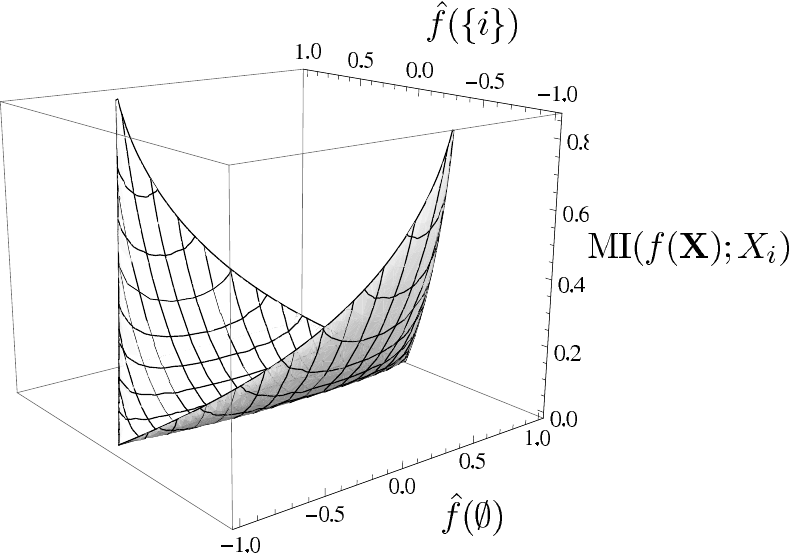}
\end{center}
\caption{$\boldsymbol{\MI{f(\vect{X})}{X_i}}$ as a function of $\boldsymbol{\PFC(\{i\})}$ and $\boldsymbol{\PFC(\emptyset)}$ for $\boldsymbol{p_i=0{.}3}$.} \label{fig2} 
\end{figure}

It can be seen that $\MI{f(\vect{X})}{X_i} = 0$, i.e., $f(\vect{X})$ and $X_i$ are statistically independent if and only if $\PFC(\{i\})=0$. That can be formalized as follows: $\MI{f(\vect{X})}{X_i}$ is convex in $\PFC(\{i\})$. This can be proven by taking the second derivative of \eqref{eq:michar} and observing that it is larger than zero for all pairs of values ($\PFC(\emptyset)$,$\PFC(\{i\})$) for which $\MI{f(\vect{X})}{X_i}$ is defined. Next, from \eqref{eq:michar}, we see that $\MI{f(\vect{X})}{X_i} = 0$ if $\PFC(\{i\})=0$; hence, it follows that $\MI{f(\vect{X})}{X_i} = 0$ if and only if $\PFC(\{i\})=0$, which proves the following result:

\begin{corollary}
Let $f$ be a Boolean function, and $\vect{X}$ be product distributed. $X_i$ and $f(\vect{X})$ are statistically independent if and only if $\PFC(\{i\})=0$.
\label{co:stat_indep_onev}
\end{corollary}

Corollary~\ref{co:stat_indep_onev} also follows immediately from a more general result, namely Theorem~\ref{le:statistical_independence_fourier_coeffs}, which is presented later. Recall that for PARITY2,  $\MI{f_{\text{PARITY2}}(\vect{X})}{X_1} = 0$ and $\PFC(\{1\})=0$; hence, Corollary~\ref{co:stat_indep_onev} comes at no surprise.

From Figure \ref{fig2}, it can be seen that the larger $|\PFC(\{i\})|$, the larger $\MI{f(\vect{X})}{X_i}$ becomes. Formally, it follows from the convexity of $\MI{f(\vect{X})}{X_i}$ and Corollary~\ref{co:stat_indep_onev} that $\MI{f(\vect{X})}{X_i}$ is increasing in $|\PFC(\{i\})|$. Hence,  $X_i$ has large determinative power, i.e., $\MI{f(\vect{X})}{X_i}$ is large, if and only if $|\PFC(\{i\})|$ is large (i.e., close to one). $|\PFC(\{i\})|$ is trivially maximized for the \textit{dictatorship function}, i.e., for $f(\vect{x})=x_i$, or its negation, i.e., $f(\vect{x})=-x_i$. The output $f(\vect{X})$ of the dictatorship function is fully determined by $x_i$.

Next, let us consider the (trivial) case where $A =[n]$ and hence $\XA=\vect{X}$. Then, $\MI{f(\vect{X})}{\vect{X}} =  h(1/2(1+ \PFC(\emptyset) )$. It follows that $\MI{f(\vect{X})}{\vect{X}}$ is maximized for $\PFC(\emptyset) \allowbreak= 0$, i.e, $\PR{f(\vect{X}) = 1} = 1/2$, i.e., if the variance of $f(\vect{X})$ is $1$. In general, the closer to zero $\PFC(\emptyset)$ is, the larger the mutual information between a function's output and all its inputs becomes. Let us finally relate the conditional entropy $H(f(\vect{X}) | \XA  )$ to the concentration of the Fourier weight on the coefficients $\{S\colon S \subseteq A\}$, $A \subseteq [n]$.

\begin{theorem}
Let $f$ be a Boolean function, let $\vect{X}$ be product distributed, and let $\XA=\{X_i\colon i\in A\}$ be a fixed set of arguments, where $A \subseteq [n]$. Then,
\begin{equation*}
\left(\! 1- \sum_{S\subseteq A} \PFC (S)^2 \! \right)^{\frac{1}{\ln(4)}}  \geq H(f(\vect{X}) | \XA ) \geq 1 - \sum_{S\subseteq A} \PFC (S)^2.
\end{equation*}
\label{th:hlowerbound}
\end{theorem}
\begin{proof}
See Appendix \ref{proof_pr7}.
\end{proof}

Theorem~\ref{th:hlowerbound} shows that $H(f(\vect{X}) | \XA )$ can be approximated with $1 - \sum_{S\subseteq A} \PFC (S)^2$. It further shows that $H(f(\vect{X}) | \XA )$ is small if the Fourier weight is concentrated on the variables in the set $A$, i.e., if $\sum_{S \subseteq A} \PFC (S)^2$ is close to one. In contrast, as mentioned previously, for $\INF{A}{f}$, it is relevant whether the Fourier weight is concentrated on the coefficients with high degree.

\subsection{Relation to measures of perturbation}

Mutual information and average sensitivity are related as follows. 

\begin{theorem}
For any Boolean function $f$, for any product distributed $\vect{X}$,
 \begin{equation}
  \INF{A}{f}  \geq   \underset{i\in A}{\min}\left( \frac{1}{\sigma_i^2}  \right) \left(\MI{f(\vect{X})}{\XA} - \Psi(\VAR{f(\vect{X})}) \right)
 \label{eq:inf_bound_mi}
  \end{equation}
  with
  \begin{equation}
  \Psi(x) \defeq   ( x)^{1/\ln(4)}-x.
  \label{eq:errorpsi}
  \end{equation}
\label{pr:inf_bound_mi}
\end{theorem}

\begin{proof}
See Appendix \ref{proof_pr5}.
\end{proof}

Note that the term $\Psi(\VAR{f(\vect{X})})$ is close to zero.  Specifically, for any $f(\vect{X})$ we have $0\leq \Psi(\VAR{f(\vect{X})}) < 0.12$, and for settings of interest, $\Psi(\VAR{f(\vect{X})})$ is very close to zero, as explained in more detail in the following. Theorem~\ref{pr:inf_bound_mi} shows that if $\MI{f(\vect{X})}{\XA}$ if large (i.e., close to one), $f$ must be sensitive to perturbations of the entries of $\XA$. Moreover, if $\INF{A}{f}$ is small (i.e., if $f$ is tolerant to perturbations of the entries of $\XA$), then $\MI{f(\vect{X})}{\XA}$ must be small (i.e., the entries of $\XA$ do not have determinative power). For the case that $A=[n]$, Theorem~\ref{pr:inf_bound_mi} states that the average sensitivity $\AV{f}$ is lower-bounded by $\MI{f(\vect{X})}{\vect{X}}$ minus some small term.

We next discuss the special case that $A = \{i\}$. Theorem~\ref{pr:inf_bound_mi} evaluated for $A = \{i\}$ yields a lower bound on the influence of a variable in terms of the mutual information of that variable, namely
\begin{equation}
  \INF{i}{f} \geq \frac{1}{\sigma_i^2}  \left(\MI{f(\vect{X})}{X_i} -  \Psi(\VAR{f(\vect{X})}) \right).
  \label{eq:imionevar}
 \end{equation}
Again, $\Psi(\VAR{f(\vect{X})})$ is close to zero for settings of interest, as the following argument explains. Equation ~\eqref{eq:imionevar} will not be evaluated for small $\VAR{f(\vect{X})}$; since then, $f(\vect{X})$ is close to a constant function (i.e., close to $f(\vect{X}) = 1$ or $f(\vect{X}) = -1$), and $ \INF{i}{f}$ and $\MI{f(\vect{X})}{X_i}$ must both be small (i.e., close to zero) anyway. Hence, \eqref{eq:imionevar} is of interest when $\VAR{f(\vect{X})}$ is large, i.e., close to 1; for this case, the term $\Psi(\VAR{f(\vect{X})}) $ is small (e.g., for $\VAR{f(\vect{X})} >0.8$, $\Psi(\VAR{f(\vect{X}})  < 0.05$). Observe that, according to \eqref{eq:imionevar}, if $\MI{f(\vect{X})}{X_i}$ is large, then $\INF{i}{f}$ is also large. That proves the intuitive idea that if an input determines $f(\vect{X})$ to some extent, this input must be sensitive to perturbations. Conversely, as mentioned previously, an input $i$ can have large influence and still $\MI{f(\vect{X})}{X_i}= 0$. E.g., for the PARITY2 function, we have $\INF{i}{f} = 1$ and $\MI{f(\vect{X})}{X_i}= 0$.

Interestingly, the influence also has an information theoretic interpretation. The following theorem generalizes Theorem 1 in \cite{diskin_voting_2010}.

\begin{theorem} For any Boolean function $f$, for any product distributed $\vect{X}$,
 \[
  \INF{i}{f} = \frac{H \!\left(f(\vect{X})| \vect{X}_{[n]\setminus\{i\}}\right)  }{H(X_i) }.
 \]
 \label{pr:inf_infortheor}
\end{theorem}

\begin{proof}
See Appendix \ref{proof_pr4}. For uniformly distributed $\vect{X}$, a proof appears in \cite{diskin_voting_2010}.
\end{proof}

Theorem~\ref{pr:inf_infortheor} shows that the influence of a variable is a measure for the uncertainty of the function's output that remains if all variables except variable $i$ are set.

\subsection{Statistical independence of inputs to a Boolean function}

Next, we characterize statistical independence of $f(\vect{X})$ and a set of its arguments $\XA$ in terms of Fourier coefficients. This result generalizes a theorem derived by Xiao and Massey \cite{xiao_spectral_1988} from uniform to product distributed $\vect{X}$.

\begin{theorem}
Let $A \subseteq [n]$ be fixed, $f$ be a Boolean function, and $\vect{X}$ be product distributed. Then, $f(\vect{X})$ and the inputs $\XA = \{X_i : i \in A \}$ are statistically independent if and only if
\[
\PFC(S) = 0   \text{ for all }  S\subseteq A \setminus \emptyset.
\]
\label{le:statistical_independence_fourier_coeffs}
\end{theorem}
\begin{proof}
See  Appendix \ref{proof_pr6}. For uniformly distributed $\vect{X}$, i.e., $\PR{X_i=1} = 1/2$ for all $i \in [n]$, Theorem~\ref{le:statistical_independence_fourier_coeffs} has been derived by Xiao and Massey \cite{xiao_spectral_1988}. Note that the proof provided here is also conceptually different from the proof for the uniform case in \cite{xiao_spectral_1988}, as it does not rely on the Xiao-Massey lemma.   
\end{proof}

Theorem~\ref{le:statistical_independence_fourier_coeffs} shows that a function and small sets of its inputs are statistically independent if the spectrum is concentrated on the coefficients of high degree $d=|S|$. The most prominent example is the parity function of $n$ variables, i.e., $f_{\mathrm{PARITYN}}(\vect{x}) = x_1x_2...x_n$: For uniformly distributed $\vect{X}$, each subset of $n-1$ or fewer arguments and $f_{\mathrm{PARITYN}}(\vect{X})$ are statistically independent. Conversely, if a function is concentrated on the coefficients of low degree $d=|S|$, which is the case for functions that are tolerant to perturbations, then small sets of inputs and the function's output are statistically dependent.

Theorem~\ref{le:statistical_independence_fourier_coeffs} also has an important implication for algorithms that detect functional dependencies in a BN based on estimating the mutual information from observations of the network's states, such as the algorithm presented in \cite{liang_reveal_1998}. Theorem~\ref{le:statistical_independence_fourier_coeffs} characterizes the classes of functions for which such an algorithm may succeed and for which it will fail. Moreover, Theorem~\ref{le:statistical_independence_fourier_coeffs} shows that in a Boolean model of a genetic regulatory network, a functional dependency between a gene and a regulator cannot be detected based on statistical dependence of a regulator $X_i$ and a gene's state $f_j(\vect{X})$, unless the regulatory functions are restricted to those for which $|\PFC(\{i\})| > 0$ holds for each relevant input $i$.

\section{Unate functions}\label{sec:unate}

In this section, we discuss unate, i.e., locally monotone functions.

\begin{definition}
 A Boolean function $f$ is said to be unate in $x_i$ if for each $\vect{x}= (x_1,...,x_n) \in \{-1,+1\}^n$ and for some fixed $a_i\in \{-1,+1\}$,
 $f(x_1,..., x_i=-a_i, ...,x_n) \leq f(x_1,..., x_i=a_i, ...,x_n)$ holds. $f$ is said to be unate if $f$ is unate in each variable $x_i$, $i\in [n]$.
 \label{def:unate}
\end{definition}

Each linear threshold function and nested canalizing function is unate. Moreover, most, if not all, regulatory interactions in a biological network are considered to be unate. That can be deduced from \cite{grefenstette_analysis_2006,raeymaekers_dynamics_2002}, and the basic argument is the following: If an element acts either as a repressor or an activator for some gene, but never as both (which is a reasonable assumption for regulatory interactions\cite{grefenstette_analysis_2006,raeymaekers_dynamics_2002}), then the function determining the gene's state is unate by definition. For unate functions, the following property holds:

\begin{proposition}
Let $f: \{-1,+1\}^n \rightarrow \{-1,+1\}$ be unate. Then,
\begin{equation}
\PFC (\{i\}) =  a_i\sigma_i \INF{i}{f}, \; \forall i \in [n],
\label{eq:condition_unate_statdep}
\end{equation}
where $a_i \in \{-1,+1\}$ is the parameter in Definition~\ref{def:unate}.
\label{pr:rel_inf_fou}
\end{proposition}

\begin{proof}
Goes along the same lines as the proof for monotone functions in Lemma 4.5 of \cite{bshouty_fourier_1996}.
\end{proof}

Note that conversely, if \eqref{eq:condition_unate_statdep} holds for each $x_i, i \in [n]$, $f$ is not necessarily unate. Inserting \eqref{eq:condition_unate_statdep} into \eqref{eq:michar} yields
\begin{equation}
 \label{eq:miinfunaterelation}
  \MI{f(\vect{X})}{X_i} = h\left( \frac{1}{2}(1+\PFC(\emptyset)) \right) - \mathbb E \left[ h\left( \frac{1}{2}\left( 1+\PFC(\emptyset) + a_i\sigma_i \INF{i}{f} \frac{X_i-\mu_i}{\sigma_i}  \right) \right)   \right],
 \end{equation}
where the expectation in \eqref{eq:miinfunaterelation} is over $X_i$. Based on \eqref{eq:miinfunaterelation}, the discussion from section~\ref{sec:reltov} on $\MI{f}{X_i}$ applies by using $\PFC (\{i\})$ and $a_i\sigma_i\INF{i}{f}$ synonymously. Hence, for unate functions, the mutual information $\MI{f}{X_i}$ is increasing in the influence $|\INF{i}{f}|$. Moreover, if $f$ is unate, and $x_i$ is a relevant variable, i.e., a variable on which the functions actually depend on, then $|\PFC (\{i\})| > 0$. From this fact and the same arguments as given in section~\ref{sec:reltov} follows:

\begin{theorem}
Let $f\colon \{-1,+1\}^n \rightarrow \{-1,+1\}$ be unate. If and only if $x_i$ is a relevant variable, then $\MI{f(\vect{X})}{X_i} \neq 0$.
\label{co:stat_dep_unate}
\end{theorem}

In a Boolean model of a biological regulatory network, this implies that if the functions in the network are unate, then a regulator and the target gene must be statistically dependent.

\section{\textit{E. coli} regulatory network}\label{sec:ecoli}

In \cite{covert_integrating_2004}, the authors presented a complex computational model of the \textit{E. coli} transcriptional regulatory network that controls central parts of the \textit{E.~coli} metabolism. The network consists of 798 nodes and 1160 edges. Of the nodes, 636 represent genes and of the remaining 162 nodes, most (103) are external metabolites. The rest are stimuli, and others are state variables such as internal metabolites. The network has a layered feed-forward structure, i.e., no feedback loops exist. The elements in the first layer can be viewed as the inputs of the system, and the elements in the following seven layers are interacting genes that represent the internal state of the system. Our experiments revealed that all functions are unate; therefore, the properties derived in section~\ref{sec:unate} apply. Note that all functions being unate is a special property of the network, since if functions are chosen uniformly at random, it is unlikely to sample a unate function, in particular if the number of inputs $n$ is large.

\subsection{Determinative nodes in the \textit{E. coli} network}
\label{sec:determinative_nodes}

We first identify the input nodes that have large determinative power (we will define what that means in a network setting shortly) and then show that a small number thereof reduces the uncertainty of the network's state significantly. Specifically, we show that on average, the entropy of the node's states conditioned on a small set of determinative input nodes, is small.

To put this result into perspective, we perform the same experiment for random networks with the same and different topology as the \textit{E. coli} network. We denote by $\vect{X} = \{X_1, ..., X_{n}\},  n = 145$ the set of inputs of the feed forward network and assume that the $X_i$ are independent and uniformly distributed. The remaining variables are denoted by $\vect{Y} = \{Y_1, ..., Y_{m} \}, m = 653$ and are a function of the inputs and the network's states, i.e.,  $Y_i=f'_i(\vect{X}, \vect{Y})$. For our analysis, the distributions of the random variables $Y_1, ..., Y_{m}$ need to be computed, since some of those variables are arguments to other functions. This can be circumvented by defining a collapsed network, i.e., a network where each state of a node is given as a function of the input nodes only, i.e., $Y_i = f_i(\vect{X})$. The collapsed network is obtained by consecutively inserting functions into each other, until each function only depends on states of nodes in the input layer, i.e., on $\vect{X}$. The collapsed network reveals the dependencies of each node on the input variables. Interestingly, in the collapsed network, it is seen that the variables chol\_xt$>$0, salicylate, 2ddglcn\_xt$>$0, mnnh$>$0, altrh$>$0, and his-l\_xt$>$0 (here, and in the following, we adopt the names from the original dataset), which appear to be inputs when considering the original \textit{E. coli} network, turn out to be not. Consider, for example, the node salicylate. The only node dependent on salicylate is mara = (( NOT arca OR NOT fnr) OR oxyr OR salicylate). However, arca = (fnr AND NOT oxyr), and it is easily seen that mara simplifies to mara = 1.

Next, we identify the determinative nodes. As argued in section~\ref{sec:mi}, $\MI{f_i(\vect{X})}{X_j}$ is a measure of the determinative power of $X_j$ over $Y_i = f_i(\vect{X})$. This motivates the definition of the determinative power of input $X_j$ over the states in the network as 
\[
D(j) \defeq \sum_{i=1}^m \MI{f_i(\vect{X})}{X_j}.
\]
Note that a small value of $D(j)$ implies that $X_j$ alone does not have large determinative power over the network's states, but $X_j$ may have large determinative power over the network states in conjunction with other variables. In principle $\sum_{i =1}^m \MI{f_i(\vect{X})}{X_{j}, X_{k}}$ can be large for some $j,k\in [n]$, even though $D(j)$ and $D(k)$ are equal to zero. This is, however, not possible in the \textit{E. coli} network since the functions are unate. Specifically, $\MI{f_i(\vect{X})}{X_{j}, X_{k}} \neq 0$ implies that $x_{j}$ or $x_{k}$ are relevant variables, and according to Theorem~\ref{co:stat_dep_unate}, $\MI{f_i(\vect{X})}{X_{j}} \neq 0$ or $\MI{f_i(\vect{X})}{X_{k}} \neq 0$. We computed $D(j)$ for each input variable and found that $D(j)$ is large just for some inputs, such as  o2\_xt (37 bit), leu-l\_xt (20.9 bit), glc-d\_xt (19.3 bit), and glcn\_xt$>$0 (17 bit), but is small for most other variables. Partly, this can be explained by the out-degree (i.e., the number of outgoing edges of a node) distribution of the input nodes. However, having a large out-degree does not necessarily result in large values of $D(j)$. In fact, in the \textit{E. coli} network, glc-d\_xt, glcn\_xt$>$0, and o2\_xt have 99, 93, and 73 outgoing edges, respectively. On the other hand, D(glc-d\_xt) = 19.3 bit and D(glcn\_xt$>$0) = 17 bit, whereas D(o2\_xt) = 37 bit.

Denote $\tau$ as a permutation on $[n]$, such that $D(X_{\tau(1)}) \geq D(X_{\tau(2)}) \geq ... \allowbreak \geq D(X_{\tau(n)})$, i.e., $\tau$ orders the input nodes in descending order in their determinative power.
We next consider $H(\vect{Y} | X_{\tau(1)},..., X_{\tau(l)}  )$ as a function of $l$ to see whether knowledge of a small set of input nodes reduces the entropy of the overall network state significantly. $H(\vect{Y} | X_{\tau(1)},..., X_{\tau(l)}  )$ has an interesting interpretation which arises as a consequence of the so called asymptotic equipartition property \cite{cover_elements_2006} (as discussed in greater detail in \cite{schober_analysis_2011}):
Consider a sequence $\vect{y}_1,..., \vect{y}_k$ of $k$ samples of the random variable $\vect{Y}$.
For $\epsilon>0$ and $k$ sufficiently large, there exists
a set $A^{(k)}_\epsilon$ of typical sequences $\vect{y}_1,..., \vect{y}_k$, such that
\begin{equation*}
    | A^{(k)}_\epsilon | \leq   2^{k (H(\vect{Y}) +\epsilon) }
\end{equation*}
and
\begin{equation*}
    \PR{ \vect{Y} \in A^{(k)}_\epsilon } > 1-\epsilon,
\end{equation*}
where $| A^{(k)}_\epsilon|$ denotes the cardinality of the set $A^{(k)}_\epsilon$. This shows that the sequences obtained as samples of $\vect{Y}$ are likely to fall in a set of size determined by the uncertainty of $\vect{Y}$. Since the output layer consists of 653 nodes, the network's state space has maximal size $2^{653}$. Since $\vect{Y}$ is a function of $\vect{X}$, $H(\vect{Y})  \leq H(\vect{X}) = 145 \text{bit}$, where for the last equality, we assume uniformly distributed inputs. Thus, without knowing the state of any input variable, the network's state is likely to be in a set of size roughly $2^{145}$. Given the knowledge about the states $X_{\tau(1)},..., X_{\tau(l)}$, the state of the network is likely to be in a set of size roughly $2^{H(\vect{Y} | X_{\tau(1)},..., X_{\tau(l)})}$. For a large network, however, $H(\vect{Y} | X_{\tau(1)},..., X_{\tau(l)})$ is expensive to compute as by definition:
\begin{equation}
H(\vect{Y} | \XA ) = \sum_{\xa} \PR{ \XA = \xa } \label{eq:Adf} \cdot \sum_{\vect{y}}  \PR{\vect{Y} = \vect{y} | \XA = \xa } \log_2  \PR{\vect{Y} = \vect{y} | \XA = \xa }.
\end{equation}
Hence, the number of terms in the sum is exponential in $n$ and $|A|$. An estimate of \eqref{eq:Adf} can be obtained by sampling uniformly at random over $\xa$ and $\vect{y}$. Instead, we will consider the following upper bound which is computationally inexpensive to compute:
\[
H(\vect{Y} | X_{\tau(1)},..., X_{\tau(l)}  )  \leq  A(l)
\]
with
\[
A(l)\defeq \!\sum_{i=1}^m H(Y_i|  X_{\tau(1)},...,X_{\tau(l)} ).
\]
The bound above follows from the chain rule for entropy \cite{cover_elements_2006}. $H(Y_i|  X_{\tau(1)},...,X_{\tau(l)} ) $ is computationally inexpensive to compute, since $Y_i$ depends on few variables only (in the \textit{E. coli} network, on $\leq 8$). For the \textit{E. coli} network, $A(l)$ is depicted in Figure \ref{fig3} as a function of $l$. Figure \ref{fig3} shows that knowledge of the states of the most determinative nodes reduces the uncertainty about the network's states significantly. In fact, the upper bound $A(l)$ is loose; hence, we even expect  $H(\vect{Y} | X_{\tau(1)},..., X_{\tau(l)})$ to lie significantly below $A(l)$. Also, note that when $A(l)$ is small, $H(Y_i | X_{\tau(1)},..., X_{\tau(l)})$ must be small on average; hence, $\PR{Y_i = 1 | X_{\tau(1)},...,X_{\tau(l)}}$ is close to one or zero on average.

\begin{figure}
\begin{center}
\includegraphics[width=0.5\textwidth]{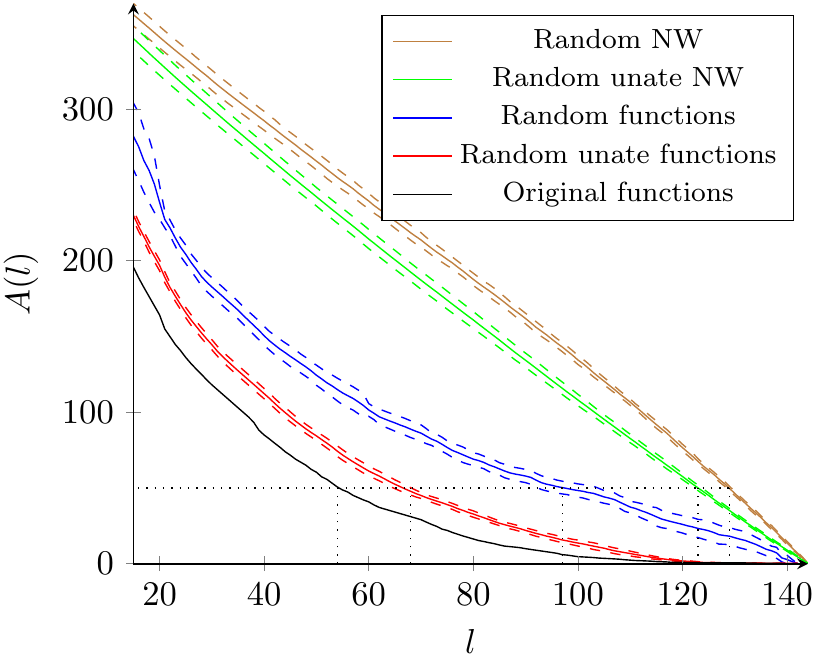}
\end{center}
\caption{The upper bound $\boldsymbol{A(l)}$ on $\boldsymbol{H(\vect{Y} | X_{\tau(1)},..., X_{\tau(l)}  )}$ as a function of $l$ for the \textit{E.~coli} network and random networks.} \label{fig3} 
\end{figure}

To put $A(l)$ for the \textit{E. coli} network in Figure \ref{fig3} into perspective, we compute $A(l)$ for random networks. First, we took the \textit{E. coli} network and exchanged each function with one chosen uniformly at random from the set of all Boolean functions of corresponding degree. We also exchanged each function with one chosen uniformly at random from all unate functions. We performed the same experiment for the original \textit{E. coli} network for 25 choices of random and random unate functions, respectively. The mean of $A(l)$, along with one standard deviation from the mean (dashed lines), is plotted in Figure \ref{fig3} for random and random unate functions. It is seen that fewer inputs determine the output of the original \textit{E. coli} network, compared to its random counterparts. For example, to obtain $A(l) = 50$, about twice as many inputs need to be known if the functions in the \textit{E. coli} network are exchanged for functions chosen uniformly at random.

Next, we generated at random feed forward networks with $m=653$ outputs and $n=145$ inputs, each with out-degree 8, i.e., the average out-degree of the inputs in the collapsed \textit{E. coli} network.
Again, we computed $A(l)$ for 25 choices of random and random unate functions, respectively.
The mean and one standard deviation from the mean are depicted in Figure \ref{fig3}. The results show that, as expected, for a random feed forward network, there seems to be no small set of inputs that determines the outputs.

\subsection{Tolerance to perturbations}
Finally, we discuss the average sensitivity of individual functions in the \textit{E. coli} network.
In section~\ref{sec:inf}, we found that the average sensitivity is small if the Fourier spectrum is concentrated on the coefficients of low degree. This appears to be the case for functions that are highly biased and for functions that depend on few variables only.
Figure \ref{fig4} shows pairs of values $(\AV{f},\text{Pr}[f(\vect{X})=1])$ for each function in the \textit{E. coli} network, again assuming that the $X_i$ are independent and uniformly distributed.
We can see from Figure \ref{fig4} that the average sensitivity of all functions is close to the lower bound on the average sensitivity. Note that the functions with high in-degree $K$ (i.e., number of relevant input variables), which could have average sensitivity up to $K$, also have small average sensitivity, because those functions are highly biased. We, therefore, can conclude that the functions have small average sensitivity either because they depend on few variables only or because they are highly biased. For other input distributions, i.e., other values of $p = \PR{X_i=1}, \forall i \in [n]$, we obtained the same results.

\begin{figure}[!h]
\begin{center}
\includegraphics[width=0.5\textwidth]{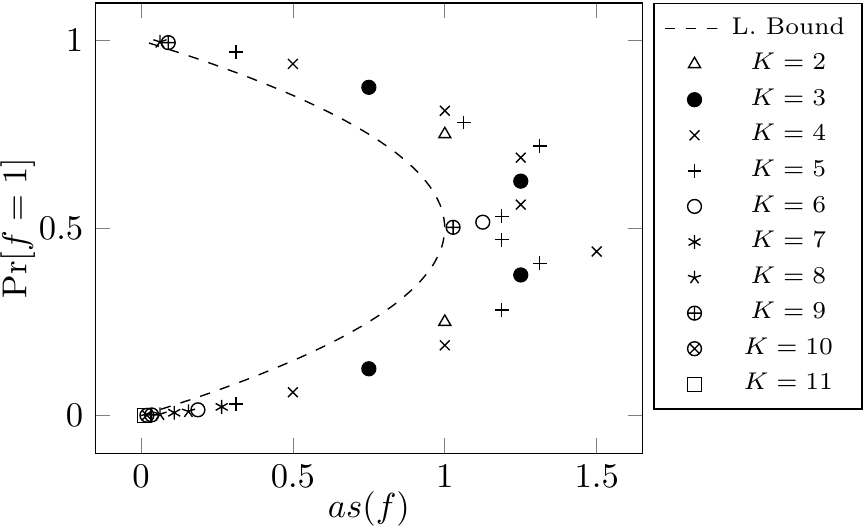}
\end{center}
\caption{Average sensitivity in the \textit{E{.} coli} network. Pairs of values $(\AV{f},\text{Pr}[f(\vect{X})=1])$ of each function in the \textit{E. coli} network for different in-degrees $K$ and uniformly distributed $\vect{X}$. Moreover, a lower bound on the average sensitivity $\AV{f}$, i.e., Poincare's inequality, is plotted.} \label{fig4} 
\end{figure}

\section{Conclusion}
In a Boolean network, tolerance to perturbations, determinative power, and statistical dependencies between nodes are properties of single functions in a probabilistic setting. Hence, we analyzed single functions with product distributed argument. We used Fourier analysis of Boolean functions to study the mutual information between a function $f(\vect{X})$ and a set of its inputs $\XA$, as a measure of determinative power of $\XA$ over $f(\vect{X})$. We related the mutual information to the Fourier spectrum and proved that the mutual information lower bounds the influence, a measure of perturbation. We also gave necessary and sufficient conditions for statistical independence of $f(\vect{X})$ and $\XA$. For the class of unate functions, which are particularly interesting for biological networks, we found that mutual information and influence are directly related (not just via an inequality). We also found that $\MI{f(\vect{X})}{X_i}>0$ for each relevant input $i$, which, as an application, implies that in a unate regulatory network, a gene and its regulator must be statistically dependent. As an application of our results, we analyzed the large-scale regulatory network of \textit{E. coli}. We identified the most determinative input nodes in the network and found that it is sufficient to know only a small subset of those in order to reduce the uncertainty of the overall network state significantly.
This, in turn, reduces the size of the state space in which the network is likely to be found significantly.

A possible direction for future work is to provide an analysis similar to that of the \textit{E. coli} regulatory network for other Boolean models of biological networks, and see if similar conclusions as in section~\ref{sec:ecoli} can be reached. One of the main assumptions in our work is the independence among the input variables of the network. It would be interesting to provide methods that can be used beyond this setup. 
However, deriving such results is challenging because for dependent inputs, the basis functions $\OBF_{S}(\vect{x})$ do not factorize as in \eqref{eq:basisfunc}, and many results cited and derived in this paper make use of this particular form of the basis functions. 
In this paper, we focused on generic properties of information-processing networks that may help identify possible principles that underly biological networks. Assessing our findings from a biological perspective would be an interesting next step.

%
%

\appendix

\section*{Appendices}

\section{\label{ap:le:relation_fourier_Evalue_general}  Lemma~\ref{le:relation_fourier_Evalue_general} }

For the proof of Theorems~\ref{th:relation_fourier_H} and ~\ref{le:statistical_independence_fourier_coeffs},
we will need the following lemma:

\newcommand{\xaa}{\xa}
\begin{lemma}
  Let $f$ be a Boolean function, let $\vect{X}$ be product distributed, and let $A\subseteq [n]$ and some fixed $\vect \xaa \in \{-1,+1\}^{|A|}$ be given.
    Then,
    \begin{equation}
    \EX{ f (\vect X) |  \vect{X}_A = \xaa } = \sum_{S \subseteq A} \PFC (S)  \OBF_{S}(\xaa).
    \label{lemma1:eq1}
    \end{equation}
    \label{le:relation_fourier_Evalue_general}
\end{lemma}

\begin{proof}
Inserting the Fourier expansion of $f(\vect{X})$ given by \eqref{eq:fourier_expansion} in the left-hand side of \eqref{lemma1:eq1} and utilizing the linearity of conditional expectation yields
\begin{equation*}
\EX{ f (\vect X) |  \vect{X}_A = \xaa } = \sum_{S\subseteq[n]} \hat f(S) ~ \EX{ \OBF_{S} (\vect X) |  \vect{X}_A = \xaa}.
\end{equation*}

For $S\subseteq A$,
\begin{equation*}
\EX{ \OBF_{S} (\vect X) | \vect{X}_A = \xaa } = \OBF_{S} (\xaa).
\end{equation*}
Conversely, for $S \not\subseteq A$,
    \[
    \EX{ \OBF_{S} (\vect X) | \vect{X}_A = \xaa } = 0.
    \]
To see this, assume without loss of generality that $S = A \cup \{j\}$ and $j \notin A$. Using the decomposition property of the basis function as given in section~\ref{sec:fanalysisintro},
\begin{align*}
\EX{ \OBF_{S} (\vect X) | \vect{X}_A = \xaa }
&= \mathbb E \left[  \prod_{i\in S} \Phi_{\{i\}}(\vect X) | \vect{X}_A = \xaa \right]  \\
& = \prod_{i\in S} \EX{ \Phi_{\{i\}}(\vect X) | \vect{X}_A = \xaa }
\end{align*}
which is equal to zero as
\begin{equation*}
\EX{ \Phi_{\{j\}}(\vect X) | \vect{X}_A = \xaa } =  \EX{ \Phi_{\{j\}}(\vect X)} = 0.
\end{equation*}
\end{proof}

\section{
\label{proof:relation_fourier_H}
Proof of Theorem~\ref{th:relation_fourier_H} }

First,
\begin{align}
&\PR{f(\vect{X})=1|\vect{X}_A = \vect{x}_A} = \frac{1}{2} \left(  1 + \EX{ f (\vect X) |  \vect{X}_A = \vect x_A} \right) \nonumber \\
&\hspace{1.2cm}=
\underbrace{
\frac{1}{2} \left(  1 +\sum_{S \subseteq A} \PFC(S) \OBF_{S}(\xa) \right)
}_{q(\xa)}, \label{eq:defqew}
\end{align}
where \eqref{eq:defqew} follows from an application of Lemma~\ref{le:relation_fourier_Evalue_general}.
By definition of the conditional entropy,
\begin{align}
H(f(\vect{X}) | \vect{X}_{A}) &= \sum_{\vect{x}_A \in \{-1,1\}^{|A|}}  \PR{\vect{X}_A = \vect{x}_A}  H ( f(\vect{X}) |  \vect{X}_A=\vect{x}_A )   \nonumber\\
&=\sum_{\vect{x}_A \in \{-1,1\}^{|A|}}  \PR{\vect{X}_A = \vect{x}_A}  h ( \PR{f(\vect{X})=1 |  \vect{X}_A=\vect{x}_A} )  \nonumber \\
&= \sum_{\vect{x}_A \in \{-1,1\}^{|A|}}  \PR{\vect{X}_A = \vect{x}_A}  h ( q(\xa) ) \label{eq:usedefeqw}\\
&= \EX{ h ( q(\XA) )  },
\label{eq:eeexpect}
\end{align}
where $h(\cdot)$ is the binary entropy function as defined in \eqref{eq:BEF}. To obtain \eqref{eq:usedefeqw}, we used \eqref{eq:defqew}. The expectation in \eqref{eq:usedefeqw} is with respect to the distribution of  $\vect{X}_A$. Inserting $q(\XA)$ as given by \eqref{eq:defqew} in \eqref{eq:eeexpect} concludes the proof.

\section{
\label{proof_pr7}
Proof of Theorem~\ref{th:hlowerbound}}

First, note that with $q(\cdot)$ as defined in \eqref{eq:defqew}, we have
\begin{align}
\EX{4q(\XA)(1-q(\XA))} &= \mathbb E \left[   1- \left( \sum_{S \subseteq A} \PFC(S) \OBF_{S}(\mathbf{\XA})\right)^2 \right] \notag \\
&= \sum_{S \subseteq A} \sum_{U \subseteq A} \PFC(S)\PFC(U) \EX{  \OBF_{S}(\XA) \OBF_{U}(\XA)} \notag \\
&= \label{eq:psumcoeff} 1- \sum_{S \subseteq A} \PFC (S)^2,
\end{align}
where \eqref{eq:psumcoeff} follows from the orthogonality of the basis functions.

We start with proving the lower bound in Theorem~\ref{th:hlowerbound}. Applying the lower bound on the binary entropy function $h(p) \geq 4  p  (1-p)$, given in Theorem 1.2 of \cite{topsoe_bounds_2001}, on \eqref{eq:eeexpect} yields
\[
 H(f(\vect{X}) | \XA ) = \EX{h(q(\XA))} \geq \EX{4  q(\XA)  (1-q(\XA))},
\]
and the lower bound in  Theorem~\ref{th:hlowerbound} follows using \eqref{eq:psumcoeff}.

Next, we prove the upper bound in Theorem~\ref{th:hlowerbound}. Applying the upper bound on the binary entropy function $h(p) \leq \left( p(1-p) \right)^{1/\ln(4)}$, given in Theorem 1.2 of \cite{topsoe_bounds_2001}, on \eqref{eq:eeexpect} yields
\begin{align}
 H(f(\vect{X}) | \XA ) &= \EX{h(q(\XA)} \nonumber\\
              &\leq \EX{ (\, \underbrace{4 q(\XA)  (1-q(\XA) )}_{Y} \, )^{1/\ln(4)}}.
\label{eq:upperboundh}
\end{align}
The term $Y$ in \eqref{eq:upperboundh} is a random variable, and the function $(Y)^{1/\ln(4)}$ is concave in $Y$. An application of Jensen's inequality (see e.g. \cite{cover_elements_2006}) yields $\EX{ (Y)^{1/\ln(4)}   } \leq ( \EX{Y} )^{1/\ln(4)}$; hence, the right-hand side of \eqref{eq:upperboundh} can be lower as
\begin{equation}
 H(f(\vect{X}) | \XA ) \leq \left( \EX{ 4 q(\XA) (1-  q(\XA) )}  \right)^{1/\ln(4)}.
\label{eq:lasteqinappe}
\end{equation}
Finally, the upper bound in Theorem~\ref{th:hlowerbound} follows from combining \eqref{eq:lasteqinappe} and \eqref{eq:psumcoeff}.

\section{
\label{proof_pr5}
Proof of Theorem~\ref{pr:inf_bound_mi}}

According to Proposition~\ref{pr:av_fourier},
 \begin{align}
  \INF{A}{f} &= \sum_{S \subseteq [n]} \PFC(S)^2 \sum_{i \in S \cap A} \frac{1}{\sigma_i^2} \nonumber \\
&\geq \sum_{S \subseteq [n] \setminus \emptyset} \PFC(S)^2 |S \cap A| \underset{i\in A}{\min} \left(\frac{1}{\sigma_i^2} \right) \nonumber \\
&\geq \underset{i\in A}{\min}\left(\frac{1}{\sigma_i^2} \right) \sum_{S \subseteq A \setminus \emptyset}  \PFC(S)^2. \label{eq:prop5qpr}
 \end{align}
Next, we rewrite the lower bound on $H(f(\vect{X}) | \XA )$ given by Theorem~\ref{th:hlowerbound} as
\begin{equation}
 \sum_{S\subseteq A \setminus \emptyset} \PFC (S)^2 \geq 1- \PFC (\emptyset)^2 -H(f(\vect{X}) | \XA ).
 \label{eq:rewrlbthmh}
\end{equation}
By adding $H(f(\vect{X}))-H(f(\vect{X}))$ on the right-hand side of \eqref{eq:rewrlbthmh} and using the definition of mutual information, \eqref{eq:rewrlbthmh} becomes
\begin{equation}
 \sum_{S\subseteq A \setminus \emptyset}  \!\! \PFC (S)^2  \! \geq \! \MI{f(\vect{X})}{\XA} - H(f(\vect{X})) + 1- \PFC (\emptyset)^2.
 \label{eq:rewrlbthmh2}
\end{equation}
With $\VAR{f(\vect{X})} = 1-\PFC(\emptyset)^2$ and by using the inequality $H(f(\vect{X})) \leq ( \VAR{f(\vect{X})})^{1/\ln(4)}$, given in Theorem 1.2 of \cite{topsoe_bounds_2001}, \eqref{eq:rewrlbthmh2} becomes
\begin{equation}
\sum_{S\subseteq A \setminus \emptyset} \PFC (S)^2 \geq \MI{f(\vect{X})}{\XA} - \Psi(\VAR{f(\vect{X})}),
\label{eq:prop5qpr2}
\end{equation}
with $\Psi(\cdot)$ as defined in \eqref{eq:errorpsi}. Finally, Theorem~\ref{pr:inf_bound_mi} follows by combining \eqref{eq:prop5qpr} and \eqref{eq:prop5qpr2}.

\section{
\label{proof_pr4}
Proof of Theorem~\ref{pr:inf_infortheor}}

For notational convenience, let $A = [n]\setminus \{i\}$. By definition of the conditional entropy,
\begin{align}
\label{eq:a1}
H(f(\vect{X}) | \vect{X}_{A}) &= \!\! \sum_{\vect{x}_A \in \{-1,1\}^{|A|}} \!\!  \PR{\vect{X}_A = \vect{x}_A}  H ( f(\vect{X}) |  \vect{X}_A=\vect{x}_A )  \nonumber \\
&= \!\! \sum_{\vect{x}_A \in \{-1,1\}^{|A|}}  \!\! \PR{\vect{X}_A = \vect{x}_A}  h ( \PR{f(\vect{X})=1 |  \vect{X}_A=\vect{x}_A}),
\end{align}
where $h(\cdot)$ is the binary entropy function as defined in \eqref{eq:BEF}.
Observe that
\[
h ( \PR{f(\vect{X})=1 |  \vect{X}_A=\vect{x}_A }) = h( \PR{X_i=1} )
\]
if
 \begin{align}
& f( X_1 =x_{1}, ..., X_{i} = 1,...,  X_{n}=x_n) \nonumber \\
& \neq f( X_{1}=x_1, ..., X_{i} = -1,..., X_n = x_{n}) \nonumber
 \end{align}
  and
 \[
 h ( \PR{f(\vect{X})=1 |  \vect{X}_A=\vect{x}_A }) =0
 \]
otherwise. Hence, \ref{eq:a1} becomes
\begin{align}
H(f(\vect{X}) | \XA) =\sum_{\vect{x}_A \in \{-1,1\}^{|A|}}  \PR{\vect{X}_A = \vect{x}_A}     h(p_i)
 \ind{f(\vect{X}) \neq f(\vect{X} \oplus e_i) }, \nonumber
\end{align}
where $\vect{x}\oplus e_i$ is the vector obtained from $\vect{x}$ by flipping its $i$th entry, and Theorem~\ref{pr:inf_infortheor} follows by using the definition of the influence.

\section{
\label{proof_pr6}
Proof of Theorem~\ref{le:statistical_independence_fourier_coeffs}}

By definition, $f(\vect{X})$ and $\vect{X}_A$ are statistically independent if and only if for all $\vect{x}_A \in \{-1,+1\}^{|A|}$
\begin{equation}
 \PR{f(\vect{X})=1|\vect{X}_A = \vect{x}_A}=\PR{f(\vect{X})=1}.
 \label{eq:co1}
\end{equation}
With
\[
 \PR{f(\vect{X})=1|\vect{X}_A = \vect{x}_A} = \frac{1}{2} + \frac{1}{2} \EX{ f (\vect X) |  \vect{X}_A = \vect x_A}
\]
and application of Lemma~\ref{le:relation_fourier_Evalue_general} given in Appendix 1, \eqref{eq:co1} becomes
\begin{align}
\sum_{S \subseteq A} \PFC (S)  \OBF_{S}
( \mathbf{\vect{x}_A })  =   \PFC (\emptyset)  \nonumber \\
\Leftrightarrow \sum_{S \subseteq A \setminus \emptyset } \PFC (S)  \OBF_{S} ( \mathbf{\vect{x}_A}) = 0.
\label{eq:statfourierindependen}
\end{align}
It follows from the Fourier expansion \eqref{eq:fourier_expansion} that \eqref{eq:statfourierindependen} holds for all $\vect{x}_A \in \{-1,+1\}^{|A|}$ if and only if $\PFC(S) = 0$ for all $S\subseteq A \setminus \emptyset$, which proves the theorem.

\section*{Acknowledgments}
We would like to thank Sara Al-Sayed and Dejan Lazich for their helpful discussions and careful reading of the manuscript.


\begin{thebibliography}{10}
\providecommand{\url}[1]{[#1]}
\providecommand{\urlprefix}{}

\bibitem{kauffman_metabolic_1969}
S. Kauffman,
Metabolic stability and epigenesis in randomly constructed genetic nets.
J. Theor. Biol.
\textbf{22}(3),
437--467
(1969)

\bibitem{kauffman_homeostasis_1969}
S. Kauffman,
Homeostasis and differentiation in random genetic control networks.
Nature.
\textbf{224}(5215),
177--178
(1969)

\bibitem{davidich_boolean_2008}
S. Davidich, M.I. Bornholdt,
Boolean network model predicts cell cycle sequence of fission yeast.
PLoS ONE.
\textbf{3}(2),
e1672
(2008)

\bibitem{saez-rodriguez_logical_2007}
J. Saez-Rodriguez, L. Simeoni, J.A. Lindquist, R. Hemenway, U. Bommhardt, B. Arndt, U Haus, R Weismantel, E.D. Gilles, S. Klamt, B. Schraven,
A logical model provides insights into T cell receptor signaling.
PLoS Comput. Biol.
\textbf{3}(8), e163 
(2007)

\bibitem{li_yeast_2004}
F. Li, T. Long, Y. Lu, Q. Ouyang, C. Tang,
The yeast cell-cycle network is robustly designed.
Proc. Natl. Acad. Sci. USA
\textbf{101}(14),
4781--4786
(2004)

\bibitem{covert_integrating_2004}
B.O. Covert, M.W. Knight, E.M. Reed, J.L. Herrgard, M.J. Palsson,
Integrating high-throughput and computational data elucidates bacterial networks.
Nature.
\textbf{429}(6987),
92--96
(2004)

\bibitem{kesseli_spectral_2005}
J. Kesseli, P. R\"{a}m\"{o}, O. Yli-Harja,
On spectral techniques in analysis of Boolean networks.
Physica D: Nonlinear Phenomena.
\textbf{206}(1--2),
49--61
(2005)

\bibitem{kesseli_tracking_2005}
J. Kesseli, P. R\"{a}m\"{o}, O. Yli-Harja,
Tracking perturbations in Boolean networks with spectral methods.
Phys Rev E
\textbf{72}(2),
026137
(2005)

\bibitem{ribeiro_mutual_2008}
L. Ribeiro, A.S. Kauffman, S.A. loyd-J Price, B. Samuelsson, J.E.S. Socolar,
Mutual information in random Boolean models of regulatory networks.
Phys Rev E.
\textbf{77},
011901
(2008)

\bibitem{heckel_determinative_2012}
R. Heckel, S. Schober, M. Bossert,
Determinative power and tolerance to perturbations in Boolean networks.
Paper presented at the 9th international workshop on computational systems biology,
Ulm, Germany,
4-6 June 2012


\bibitem{xiao_spectral_1988}
G. Xiao, J. Massey,
A spectral characterization of correlation-immune combining functions.
Inf. Theory IEEE Trans.
\textbf{34}(3),
569--571
(1988)

\bibitem{liang_reveal_1998}
S. Liang, S. Fuhrman, R. Somogyi,
Reveal, a general reverse engineering algorithm for inference of genetic network architectures.
Pacific Symposium on Biocomputing
\textbf{3}
18--29
(1998)

\bibitem{grefenstette_analysis_2006}
J. Grefenstette, S. Kim, S. Kauffman,
An analysis of the class of gene regulatory functions implied by a biochemical model.
Bio. Syst.
\textbf{84}(2),
81--90
(2006)

\bibitem{derrida_random_1986}
B. Derrida, Y. Pomeau,
Random networks of automata: a simple annealed approximation.
Europhysics. Lett. \textbf{1}(2),
45--49
(1986)

\bibitem{bahadur_representation_1961}
R.R. Bahadur,
A representation of the joint distribution of responses to n dichotomous items.
in \textit{Studies in Item Analysis and Prediction},
ed. by H. Solomon
(Stanford University Press, Stanford, 1961),
pp.158--168 

\bibitem{ben-or_collective_1985}
M. Ben-Or, N. Linial,
Collective coin flipping, robust voting schemes and minima of Banzhaf values.
Paper presented at the 26th annual symposium on foundations of computer science,
Portland, Oregon, USA,
21--23 October 1985

\bibitem{bshouty_fourier_1996}
C. Bshouty, N.H. Tamon,
On the Fourier spectrum of monotone functions.
J. ACM.
\textbf{43}(4),
747--770
(1996)

\bibitem{schober_about_2008}
S. Schober,
About Boolean networks with noisy inputs.
Paper presented at the fifth international workshop on computational systems biology,
Leipzig, Germany,
11--13 June 2008 

\bibitem{matache_sensitivity_2009}
V. Matache, M.T. Matache,
On the sensitivity to noise of a Boolean function.
J. Math. Phys.
\textbf{50}(10),
103512
(2009)

\bibitem{cover_elements_2006}
J.A. Cover, T.M. Thomas,
\textit{Elements of Information Theory},
2nd edn. (Wiley-Interscience, New York, 2006)

\bibitem{forre_methods_1990}
R. Forre,
Methods and instruments for designing S-boxes.
J. Cryptology
\textbf{2}(3), 115-130
(1990)

\bibitem{rosell_why_2005}
B. Rosell, L. Hellerstein, S. Ray,
Why skewing works: learning difficult Boolean functions with greedy tree learners.
Paper presented at the 22nd international conference on machine learning,
Bonn, Germany,
7--11 August 2005


\bibitem{diskin_voting_2010}
A. Diskin, M. Koppel,
Voting power: an information theory approach.
Soc. Choice Welfare
\textbf{34},
105--119
(2010)

\bibitem{raeymaekers_dynamics_2002}
L. Raeymaekers,
Dynamics of Boolean networks controlled by biologically meaningful functions.
J. Theor. Biol.
\textbf{218}(3),
331--341
(2002)

\bibitem{schober_analysis_2011}
S. Schober,
\textit{Analysis and Identification of Boolean Networks using Harmonic Analysis},
(Der Andere Verlag, Germany, 2011) 


\bibitem{topsoe_bounds_2001}
F. Topsoe,
Bounds for entropy and divergence for distributions over a two-element set.
Inequalities Appl. Pure Math.
\textbf{2},
(2001)

\end{thebibliography}
\end{document}